\documentclass[aps,prx,showpacs,twocolumn,reprint]{revtex4-2}

\usepackage{qcircuit}
\usepackage[dvips]{graphicx}
\usepackage{amsmath,amssymb,amsthm,mathrsfs,amsfonts,dsfont}
\usepackage{subfigure, epsfig}
\usepackage{braket}
\usepackage{hyperref}
\usepackage{bm}
\usepackage{enumerate}
\usepackage{color}
\usepackage{graphicx}
\usepackage{pgf}

\usepackage{amsthm}

\newtheorem{result}{Result}
\newtheorem{lemma}{Lemma}
\newtheorem{theorem}{Theorem}

\newcommand{\tr}{\mathrm{Tr}}

\newcommand{\qf}{\mathbf{F}_Q}
\newcommand{\finv}{\tilde{\mathbf{F}}_Q^{-1}}

\newcommand{\im}{\mathrm{Im}}
\newcommand{\re}{\mathrm{Re}}
\newcommand{\var}{\mathrm{Var}}
\newcommand{\spec}{\mathrm{Spc}}
\newcommand{\cond}{\mathrm{Cnd}}
\newcommand{\samplf}{N_{F}}
\newcommand{\samplg}{N_{g}}
\newcommand{\samplgd}{N_{smpl}}
\newcommand{\samplfmat}{\mathbf{N}_{F}}
\newcommand{\samplgvec}{\mathbf{N}_{g}}
\newcommand{\nopt}{N_\mathrm{opt}}
\newcommand{\hamilnorm}{\spec[\mathcal{H}]}
\newcommand{\fF}{f_F}
\newcommand{\fg}{f_g}
\newcommand{\summplmat}{Appendix}

\begin{document}


\title{Measurement cost of metric-aware variational quantum algorithms}

\author{Barnaby van Straaten}
\author{B\'alint Koczor}

\email{balint.koczor@materials.ox.ac.uk}

\affiliation{Department of Materials, University of Oxford, Parks Road, Oxford OX1 3PH, United Kingdom}


\begin{abstract}
We consider metric-aware quantum algorithms which use a quantum computer to efficiently estimate both a matrix and a vector object. For example, the recently introduced quantum natural gradient approach uses the Fisher matrix as a metric tensor to correct the gradient vector for the co-dependence of the circuit parameters. We rigorously characterise and upper bound the number of measurements required to determine an iteration step to a fixed precision, and propose a general approach for optimally distributing samples between matrix and vector entries. Finally, we establish that the number of circuit repetitions needed for estimating the quantum Fisher information matrix is asymptotically negligible for an increasing number of iterations and qubits.
\end{abstract}

\maketitle

\section{Introduction}

With quantum computers rising as realistic technologies, attention has turned to how such machines could perform as variational tools~\cite{farhi2014quantum,peruzzo2014variational,wang2015quantum,PRXH2,PhysRevA.95.020501,mcclean2016theory,
	PhysRevLett.118.100503,Li2017,PhysRevX.8.011021,Santagatieaap9646,kandala2017hardware,kandala2018extending,
	PhysRevX.8.031022,romero2017strategies,higgott2018variational,SuguruExc, mcclean2017hybrid,colless2017robust,kokail2018self,sharma2020noise,cerezo2020variational,SuguruGeneral, koczor2020quantum, koczor2020exponential}.
This results in a hybrid model with an iterative loop: a classical processor determines how to update the parameters describing a
family of quantum states (parametrised ansatz states), while a quantum coprocessor generates and performs measurements on that state
(via an ansatz circuit). This is of particular interest in the context of noisy, intermediate-scale quantum devices
(NISQ devices)~\cite{preskill2018quantum}, because complex ansatz states can be prepared with shallow  circuits~\cite{kassal2011simulating,C2CP23700H,whaley2014quantum,ourReview}.
Such shallow circuits will potentially enable obtaining useful value
before the era of resource-intensive quantum fault tolerance methods.
As such, variational quantum algorithms promise to solve key problems that are intractable to classical computers, such as finding ground
states \cite{peruzzo2014variational,PRXH2,mcclean2016theory,kandala2017hardware,google2020hartree}---as relevant
in quantum chemistry and in materials science---or approximately solving combinatorial problems \cite{farhi2014quantum} and beyond.

Despite their potential power, variational algorithms might require an extremely large
	number of quantum-circuit repetitions -- optimally using quantum resources will therefore have a crucial economic importance.
Attention has recently been focused on statistical aspects
of these variational quantum algorithms
\cite{sweke2019stochastic,kubler2019adaptive,qfi,Crawford2019,arrasmith2020operator,hadfield2020measurements},
such as the effect of shot noise and the reduction of their measurement costs.
It is our aim in this work to establish general scaling results by rigorously characterising
the number of measurements required to obtain a single iteration step
in case of so-called metric-aware quantum algorithms. Let us first
introduce basic notions.

\subsection{Variational quantum algorithms}
We consider variational quantum algorithms which typically aim to prepare
a parametrised quantum state $\rho(\underline{\theta})  := \Phi(\underline{\theta}) \, \rho_0$
where we model via a mapping $\Phi(\underline{\theta})$ that acts 
on the computational zero state $\rho_0$ of $N$ qubits
and depends continuously on the parameters $\theta_i$ with $i\in\{1, 2, \dots \nu \}$.
This mapping can in general
contain non-unitary elements, such as measurements \cite{koczor2019quantum, PhysRevLett.126.220501},
but in many applications one assumes that it acts (approximately) as a unitary circuit
that decomposes into a product of individual quantum gates. These gates typically
act on a small subset of the system, e.g., one and two-qubit gates.

Recently a novel variational algorithm was proposed for simulating real-time quantum evolution
using shallow quantum circuits \cite{Li2017} 
and was further generalised to imaginary time and natural gradient evolutions
\cite{samimagtime,koczor2019quantum} which can be used as optimisers of variational quantum eigensolvers (VQE)
	\cite{peruzzo2014variational,PRXH2,Rebentrost_2019, SuguruExc}.
This was shown to significantly outperform other approaches, such as simple gradient descent, in terms
of convergence speed and accuracy according to numerical simulations \cite{samimagtime,koczor2019quantum,wierichs2020avoiding}.

In this work, we consider generalisations of the aforementioned techniques as variational algorithms that
need to estimate the following two objects:
(a) a positive-semidefinite, symmetric matrix, which is usually the quantum Fisher information
that characterises sensitivity with respect to parameters $\theta_k$;
(b) a vector object that is in many applications the gradient vector of the loss function.
Examples of such algorithms are provided in references~\cite{li2017efficient,xiaotheory,samimagtime,koczor2019quantum,quantumnatgrad},
and we will refer to them in the following as metric-aware quantum algorithms.
The metric tensor typically only depends on the
parameter values while the vector object additionally depends on, e.g., a Hermitian observable
$\mathcal{H}$ that in typical scearios represents the Hamiltonian of a physical system
and decomposes into a polynomially increasing number $r_h$ of Pauli terms.

\subsection{Quantum natural gradient}
To be more concrete, in the following we will focus on one prominent algorithm, the recently
introduced quantum natural gradient approach \cite{koczor2019quantum,quantumnatgrad}
which is equivalent to imaginary time evolution when quantum circuits are noiseless and unitary \cite{koczor2019quantum,samimagtime}.
This approach can be used as a VQE optimiser when minimising the
expectation value $E(\underline{\theta}) :=\tr[ \rho(\underline{\theta}) \mathcal{H} ]$
over the parameters $\underline{\theta}$.
However, the approach generalises to any Lipschitz continuous mapping as an objective function \cite{koczor2019quantum}.

In particular, natural gradient descent governs the evolution of the ansatz parameters 
according to the update rule  \cite{koczor2019quantum}	
\begin{equation} \label{naturalgradEvoRESULUT}
\underline{\theta}(t{+}1) = \underline{\theta}(t) - \lambda \, \qf^{-1} \underline{g},
\end{equation}
where $t$ is an index and $\lambda$ is a step size.
Here the inverse of the positive-semidefinite, symmetric quantum Fisher information matrix $\qf \in \mathbb{R}^{\nu \times \nu}$
corrects the gradient vector $g_k := \partial_k E(\underline{\theta})$ for the co-dependence 
of the parameters, and both
objects can be estimated efficiently using a quantum computer while the inverse
$\qf^{-1}$  is computed by a classical processor.

We discuss different protocols for estimating the matrix $[\qf]_{kl}$ and vector $g_k$
entries for both pure (idealised, perfect quantum gates) and mixed quantum states
(via imperfect quantum gates or non-unitary elements as measurements) in the \summplmat.
We now highlight two results.
a) We derive the general upper bound $[\qf]_{kl} \leq r_g^2$, where $r_g$ is the maximal
number of Pauli terms into which generators of ansatz gates can be decomposed (Lemma~1). 
This bound is a generalisation of what is known
as the Heisenberg limit in quantum metrology \footnote{Where the ansatz parameter $\theta$ corresponds
	to a global $Z$ rotation of all the qubits and therefore $r_g = N$.
}, refer also to \cite{review,giovannetti11,koczor2019variational}.
b) The matrix $\qf$ might be ill-conditioned and the inversion in Eq.~\ref{naturalgradEvoRESULUT}
requires a regularisation.
We will use the simple variant of Tikhonov regularisation $\finv := [\qf {+} \eta \mathrm{Id}]^{-1}$
in the following; we derive analytical lower and upper bounds on the singular values
of this inverse matrix in the \summplmat~(Lemma~3) using a).

\section{Upper bounds on the measurement cost}
 To motivate our approach, we illustrate in Fig~\ref{fig1} (a/green)
how naively using the same number of measurements for estimating each matrix and vector
entry, such as in \cite{wierichs2020avoiding}, can result in impractical sampling costs.

In particular, we aim to reduce the error due to shot noise (finite sampling) $\epsilon$ of the vector
$\underline{v} := \finv \underline{g}$ in the update rule in Eq.~\eqref{naturalgradEvoRESULUT}.
We first express how the error in the matrix and vector entries
propagates to the parameter-update rule in Eq.~\eqref{naturalgradEvoRESULUT}.
We quantify this error as the expected Euclidean distance
$\langle \lVert \Delta v \rVert^2 \rangle  = \epsilon^2$, and
this translates to the condition $\sum_{k=1}^\nu \var[v_k]  = \epsilon^2$,
where $\var[v_k]$ is the variance of a single vector entry.

We derive an analytical formula in Lemma~2 in the \summplmat:
we express the error $\epsilon$ in terms of the variances $\var\{[\qf]_{k l}\}$ and $\var[g_l]$
of the measurements used to estimate
the matrix and vector entries, respectively, as
\begin{equation} \label{errorpropagation}
\epsilon^2 =  \sum_{k, l=1}^\nu a_{k l} \var\{[\qf]_{k l}\}
+ \sum_{k=1}^\nu b_k \var[g_k].
\end{equation}
The coefficients $a_{kl}$ and $b_{k}$ describe how the error of $[\qf]_{k l}$ and $g_k$ propagates through
matrix inversion and subsequent vector multiplication into the precision $\epsilon$.
We remark that these results are completely general and can be applied to any
quantum algorithm that requires the estimation of both an inverse matrix and a vector object,
such as a Hessian-based optimisation.

\addtocounter{footnote}{1}
\footnotetext[\value{footnote}]{
		In Fig.~1 in the uniformly distributed scenario the same number
		of measurements are used to determine every entry of the metric tensor (gradient
		vector) and only the overall number $\samplf$ ($\samplg$) of measurements to determine the
		metric tensor (gradient vector) is chosen optimally. In the naive scheme
		both $\samplf$  and $\samplg$ are additionally fixed
}
\newcounter{footcombined}
\setcounter{footcombined}{\value{footnote}}

\begin{figure*}[tb]
	\begin{centering}
		\includegraphics[width=\textwidth]{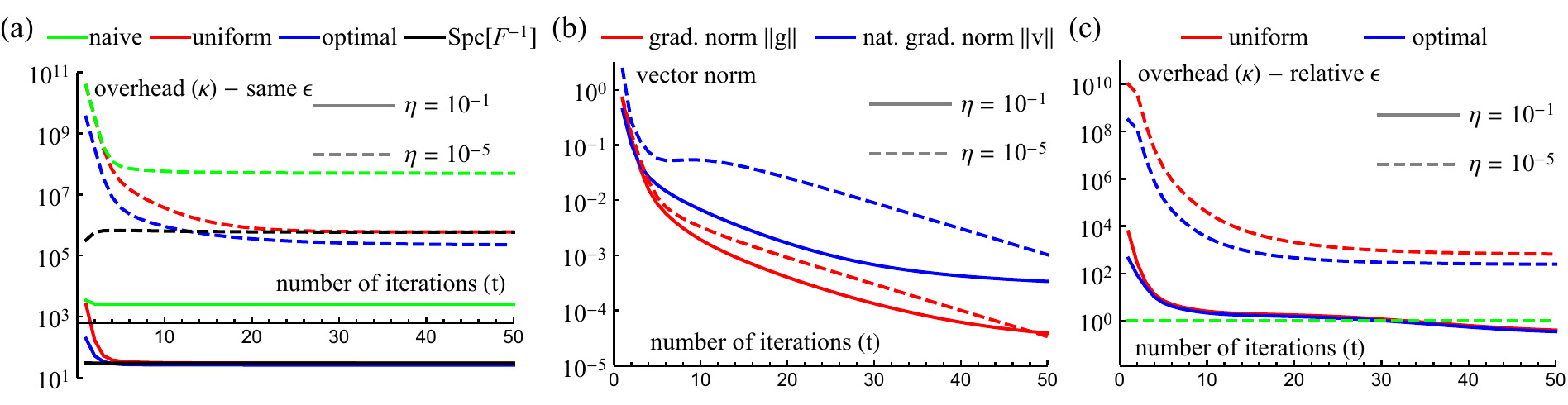}
		\caption{
			Exact numerical simulations: a $12$-qubit ansatz circuit with 84 parameters is
			initialised at a good approximation of the ground state
			of a spin-chain Hamiltonian (refer to \summplmat).
			Natural gradient evolution from Eq.~\eqref{naturalgradEvoRESULUT} was simulated with a regularisation parameter
			$\eta = 10^{-1}$ ($10^{-5}$), see solid (dashed) lines.
			(a) Measurement overhead $\kappa$ from Result~\ref{result1} (red~\cite{Note\thefootcombined}) at every iteration step $t$ of the natural
			gradient evolution. 
			This quantifies how much more it costs to estimate the natural gradient vector $\underline{v}(t)$
			than it would cost to estimate the gradient vector $\underline{g}(t)$ assuming the same precision $\epsilon$.
			$\kappa$ converges to its constant (black) asymptotic approximation.
			Optimally distributing measurements (blue) via Result~\ref{result3} significantly reduces sampling costs.
			However, naively (green~\cite{Note\thefootcombined}) using a fixed number of measurements for estimating each matrix and vector element
			results in a substantial overhead.
			(b) Multiplying $\underline{g}(t)$ (red) with the inverse of $\qf$ results in
			$\underline{v}(t)$ (blue) whose norm might be orders of magnitude larger.
			(c) In practice a relative precision is required, such that $\epsilon$ is proportional to the vector norms,
			refer to text.
			Carefully setting the regularisation parameter $\eta$ significantly improves the practical applicability:
			solid lines with $\eta = 10^{-1}$
			result in a sampling cost of $\underline{v}(t)$ comparable to (green shows $\kappa=1$) or even
			smaller than $\underline{g}(t)$.
			Refer to the main text for a remark about mitigating the initial high overheads seen in graphs (a) and (c).
			\label{fig1}
		}
	\end{centering}
\end{figure*}

We derive general upper bounds on the variances $\var\{[\qf]_{k l}\}$ and $\var[g_l]$ for different
experimental strategies in the \summplmat;
The error $\epsilon^{2}$ in Eq.~\ref{errorpropagation} is reduced proportionally when repeating measurements.
In the following, we assume that $\samplf$ measurements are assigned to estimate the full matrix $\qf$
while $\samplg$ measurements are used to estimate the gradient vector $\underline{g}$~\cite{Note\thefootcombined}.
We now state an upper bound on them in terms of the precision $\epsilon$.
\begin{theorem} \label{theo1}
	To reduce the uncertainty of the vector $\underline{v} = \finv \underline{g}$
	due to shot noise to a precision $\epsilon$,
	the number of samples to
	estimate the matrix $\qf$ in Eq.~\eqref{naturalgradEvoRESULUT} is upper bounded
	as 
	\begin{equation} \label{nfupperbound}
		\samplf  \leq   2 \,  \epsilon^{-2}  \nu^4 \, \spec[\finv]^2 \,  \lVert g \rVert_\infty^2 	 \, \fF
	\end{equation}
	while sampling the	gradient has a cost upper bounded by
 	\begin{equation}\label{thqeq2}
	 \samplg \leq 2 \epsilon^{-2} \,  \nu^2 \, \spec[\finv]  \,  \hamilnorm \fg.
	\end{equation}
	The overall measurement cost of determining the natural gradient vector is $\samplf {+} \samplg$.
	Here $\spec[A]$ denotes the average squared singular values of a matrix $A \in \mathbb{C}^{d \times d}$
	via its Hilbert-Schmidt or Frobenius norm as $\spec[A] := \lVert A \rVert^2/d$
	and  $\lVert g \rVert_\infty$ is the absolute largest
	entry in the gradient vector. 
\end{theorem}

The constant factors $\fF$ and $\fg$ in Theorem~\ref{theo1} are specific to the experimental setup
used to estimate the matrix or vector entries.
For example, for $r_g = 1$ the factor simplifies as $\fF \leq 2$.
The upper bounds in Theorem~\ref{theo1} crucially depend on the regularisation
and we prove that $\spec[\finv] \leq \eta^{-2}$, refer to Lemma~3 in the \summplmat.
The product $\hamilnorm \fg$ is a constant that reflects the complexity
of estimating the expected value of the fixed $\mathcal{H}$ (and can be reduced with advanced
techniques that simultaneously estimate commuting terms \cite{Crawford2019,yen2020measuring,jena2019pauli,gokhale2020n,gokhale2019minimizing,hadfield2020measurements}).
It is interesting to note that the sampling cost of the gradient vector $\samplg$
depends on the metric tensor via $\spec[\finv]$ (and vice versa). Let us illustrate this point
in an example where one of the entries in $\finv$
is extremely large in absolute value and therefore via the matrix/vector product it magnifies both the mean
and the variance of the gradient entries. Indeed, reducing such a magnified variance to our fixed precision
$\epsilon$ requires
an increased number of measurements in the gradient vector.
We finally remark that Theorem~\ref{theo1} is quite general and
the upper bounds apply to all metric-aware quantum algorithms
\cite{li2017efficient,xiaotheory,samimagtime,koczor2019quantum,quantumnatgrad}
up to minor modifications.

We will establish in the following, that in many cases sampling the gradient vector $\samplg$ dominates
the overall cost of the natural gradient approach as $\samplf{ +} \samplg \approx \samplg$.
Before doing so, 
let us first bound the sampling cost of the natural gradient vector \emph{relative to}
the sampling cost $\samplgd$ of the gradient vector that would be used in simple gradient descent
optimisations. Note that the difference between $\samplg$ and $\samplgd$
is that the latter corresponds to the scenario when we fix the metric tensor as
the identity matrix $F_Q:=\mathrm{Id_\nu}$
and thus the precision is $ \epsilon^2 := \langle \lVert \Delta g \rVert^2 \rangle $.

\begin{theorem} \label{theo2}
	Determining the natural gradient
	vector to the same precision $\epsilon$ as the gradient vector
	requires a sampling overhead $\kappa:= (\samplf{ +} \samplg)/\samplgd$.
	This overhead is upper bounded in general
	\begin{equation*}
		\kappa \leq \eta^{-2} +y , \quad \text{and }\quad 
		\kappa \approx \spec[\finv] +y,
	\end{equation*}
	up to the potentially vanishing term $y=\samplf / \samplgd$,
	as in Result~\ref{result1} and Result~\ref{result2}.
	Here  $\eta$ is either a regularisation parameter 
	or the smallest singular value of $\qf$.
	The second equality establishes
	an approximation as a constant factor which is valid,
	e.g.,
	when the evolution is close to the optimal point.
\end{theorem}

\section{Scaling as a function of the iterations}
Theorem~\ref{theo1} establishes that the sampling cost $\samplf$ of the matrix $\qf$ depends
on the norm of the gradient vector, which is expected to decrease polynomially during an optimisation.
In a typical scenario we expect that, even if initially estimating the matrix dominates
the sampling costs, asymptotically 
sampling the vector $\underline{g}$ dominates the costs.
\begin{result}\label{result1}
	The upper bound in Theorem~\ref{theo1} results in the growth rate
	$\samplf + \samplg = \mathcal{O}(\lVert g(t) \rVert_\infty^2) + \samplg$
	when viewed as a function of iterations or steps $t$.
	Assuming polynomial convergence via $\lVert g(t) \rVert_\infty= \mathcal{O}( t^{-c})$
	with $c>0$, the natural gradient vector requires only a constant sampling overhead
	 asymptotically as
	\begin{equation*}
	\kappa =	(\samplf{+}\samplg)/\samplgd =  \mathcal{O}(\spec[\finv] +  t^{-2c} ),
	\end{equation*}
	when compared to the gradient vector via Theorem~\ref{theo2}.
	We remark that convergence is guaranteed under mild continuity conditions \cite{sweke2019stochastic}.
\end{result}

We have numerically simulated the natural gradient evolution from Eq.~\eqref{naturalgradEvoRESULUT}
and determined its overhead $\kappa$.
This quantifies how much more it costs
at every iteration step $t$ to estimate the natural gradient vector 
$\underline{v}(t)$ than it would cost to estimate the gradient vector
$\underline{g}(t)$ assuming the same precision $\epsilon$.
Fig~\ref{fig1} (a/red) shows how this sampling overhead converges to its constant asymptotic
approximation as the average squared singular values $\spec[\finv]  \approx 10^{6}$ ($10^{1.5}$)
in Fig~\ref{fig1} (a/black).
Fig~\ref{fig1} (a/dashed) also demonstrates that under-regularising the inverse (via $\eta = 10^{-5}$)
results in unfeasible sampling costs.
In fact, carefully increasing the regularisation parameter (as $\eta = 10^{-1}$) reduces the sampling cost by
several orders of magnitude without significantly affecting the performance:
both evolutions decrease the gradient norm with a similar rate,
compare solid and dashed red lines in Fig~\ref{fig1} (b).

It is striking that the overhead plotted in Fig~\ref{fig1} (a) can be very high initially;
while the focus of the present paper is on the asymptotic costs with respect to time and size,
it is worth noting that this high initial cost could be straightforwardly
mitigated by, e.g., only occasionally updating a low-rank approximation of the metric tensor.
This may be expected to have little impact on the convergence rate since in the early
phase the advantage of using natural gradient
is typically less pronounced.

\begin{figure*}[tb]
	\centering	
	\includegraphics[width=\textwidth]{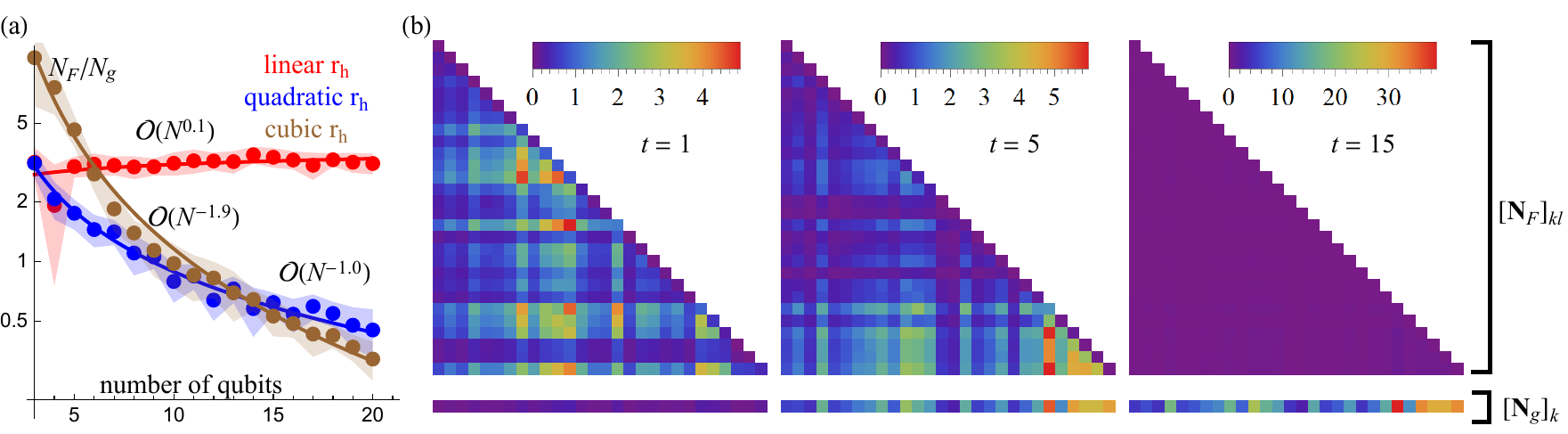} 	
	\caption{(a) Sampling cost $\samplf$ of the matrix $\qf$ relative
		to the sampling cost $\samplg$ of the gradient vector $\underline{g}$
		from Theorem~\ref{theo1} for an increasing number of qubits determined
		at randomly selected ansatz parameters with a fixed $\lVert \underline{g} \rVert = 0.1$.
		The relative sampling cost $\samplf/\samplg$
		vanishes asymptotically if the number $r_h$ of Pauli terms grows quadratically (blue), cubically (brown) or beyond as established in Eqs.~(\ref{qubitscaling1}-\ref{qubitscaling2}),
		refer also to Result~\ref{result2}.
		Shading represents the standard deviation, refer to \summplmat.
		(b)
		Color-maps showing the optimal number of measurements, $[\samplfmat]_{kl}$ and $[\samplgvec]_{k}$
		assigned to individual elements of the Fisher matrix and gradient respectively after  $1,5, 15$ iterations.
		For clarity here we have scaled the data by a constant factor such that if measurements were
		distributed uniformly than every matrix and vector element would receive $1$ measurement,
		i.e, we set the total number of measurements $\nopt$ to be the number of degrees of freedom. Most of the measurements are assigned
		to the gradient vector for an increasing number of iterations, as established in Result~\ref{result1}. 
	}	
	\label{fig2}
\end{figure*}

Recall that Fig~\ref{fig1} (a) via Result~\ref{result1} assumes a constant precision $\epsilon$
throughout the evolution which is not practical.
In fact, one would require a relative precision
such that $\epsilon =  \epsilon_0 \lVert  \underline{g}(t) \rVert $ in case of
the gradient vector and $\epsilon  =  \epsilon_0  \lVert  \underline{v}(t) \rVert $ 
in case of the natural gradient vector, for some fixed  $\epsilon_0$.
In particular, using a moderate regularisation of the inverse as $\eta = 0.1$,
the cost of estimating $\underline{v}(t)$ is comparable or even smaller
than estimating $\underline{g}(t)$, see [Fig~\ref{fig1} (c/red) solid].

We finally stress that in Fig.~\ref{fig1}(a) we do not actually compare the overall performance of
the simple and natural gradient methods, but only their per-iteration (per-epoch) costs. 
We therefore conclude that
the natural gradient
optimisation requires overall less samples to converge (i.e., asymptotically constant overhead but
faster convergence rate) when compared to simple gradient descent, see also \cite{samimagtime,koczor2019quantum,wierichs2020avoiding}. 
Moreover, we prove in the following that even the significant initial overheads in Fig~\ref{fig1} (a-c)
do in many practical applications  asymptotically vanish for an increasing number of qubits.

\section{Scaling with the system size}
Let us now consider how the upper bounds in Theorem~\ref{theo1}
scale with the number of qubits $N$.
First, we consider the general growth rate $\nu = \mathcal{O}(N a(N))$ of the number of parameters $\nu$ where
$a(N)$ is the depth of the ansatz circuit.
For example, $\mathrm{polylog}(N)$-depth circuits constitute a very general
class of ans{\"a}tze via $a(N) = \mathcal{O} (x \log(N)^y)$ for some $x,y>0$.
Second, we establish that 
the spectral quantity scales with the number of qubits as
$\spec[\finv] = \mathcal{O}(N^{-s} a^{-s}(N))$ with $0 \leq s \leq 2$, refer to Lemma~3 in the \summplmat.
Third, if $\mathcal{H}$ decomposes into a number $r_h$ of
	Pauli terms that grows polynomially (e.g., $N^4$ in case of chemistry applications) then
	we obtain a polynomial growth rate $\hamilnorm \fg = \mathcal{O}(N^b)$
	in Theorem~\ref{theo1} (Eq.~\eqref{thqeq2}) with some $b\geq1$.  We finally obtain the growth rates
\begin{align} \label{qubitscaling1}
\samplf  &= \mathcal{O}[   N^{4-2s} \, a^{4-2s}(N)  \,
\, \lVert g \rVert_\infty^2  ], \\
\samplg  &= \mathcal{O}[   N^{2-s+b} \,  a^{2-s}(N)] .
\label{qubitscaling2}
\end{align}
Note that the vector norm $\lVert g \rVert_\infty^2$ might in general also depend on the number of qubits,
e.g., exponentially vanishing gradients in case of barren plateaus~\cite{mcclean2018barren,grant2019initialization,cerezo2020cost}
which would result in an exponentially decreasing relative sampling cost of the metric tensor.
One may also think of scenarios where the gradient norm grows, however, one could then in practice decrease the inverse
precision $\epsilon^{-1}$ proportionally as typical at the initial stages of an optimisation.
To simplify our discussion, we assume that the gradient norm $\lVert g \rVert_\infty$
is fixed (bounded), e.g., the evolution is initialised
in a close vicinity of the optimal parameters as a good classical guess is known.
This also encompasses scenarios where the optimisation is near termination approaching
a fixed e.g., chemical precision.
We summarise the resulting measurement cost in the following.

\begin{result}\label{result2}
Assume that the number of Pauli terms in the Hamiltonian grows polynomially
implying $\hamilnorm \fg = \mathcal{O}(N^b)$ for some $b\geq 1$, and the gradient norm
$\lVert g \rVert_\infty = \mathcal{O}(N^{1})$ is bounded. 
The relative sampling cost of the matrix $\qf$ vanishes for general $\mathrm{polylog}(N)$-depth circuits
when $b > (2{-}s)$ and, following Theorem~\ref{theo2}, determining the natural gradient vector
requires at most a constant overhead asymptotically
\begin{equation*}
\kappa = (\samplf{+}\samplg)/\samplgd =  \mathcal{O}( \spec[\finv] + N^{2-b}  ),
\end{equation*}
when compared to the gradient vector.
\end{result}

Note that Result~\ref{result2} guarantees
a vanishing sampling cost of the matrix $\qf$
when the number of terms in the Hamiltonian
grows faster than quadratically, i.e.,  $b > 2$.
We have explicitly calculated the growth rates $b$
in case of 3 example Hamiltonians in the \summplmat\, as $b=1,2,3$, respectively,
and plot the relative sampling costs $\samplf/\samplg$ in Fig~\ref{fig2}.
We remark that this result can be applied to the
general class of metric-aware quantum
algorithms \cite{li2017efficient,xiaotheory,samimagtime,koczor2019quantum,quantumnatgrad}.

\section{Optimal measurement distribution}
So far we have assumed that $\samplf$ ($\samplg$) measurements are distributed uniformly 
among the $\nu^2$ ($\nu$) matrix (vector) entries~\cite{Note\thefootcombined}.
However, the overall number of samples $\samplf + \samplg$ (from Theorem~\ref{theo1}),
needed to obtain the vector $\underline{v} = \finv \underline{g} $ to a precision $\epsilon$,
can be minimised by distributing samples between the elements of $\qf$ and $\underline{g}$ optimally \cite{Crawford2019}.  We denote the matrix $[\samplfmat]_{kl}$ and the vector $[\samplgvec]_k$ entries that represent the number of measurements assigned to individual elements
in $\qf$ and in $\underline{g}$, respectively.
The number of samples required is reduced to
$\nopt = \Sigma^{2} / \epsilon^{2}$ with $\nopt \leq \samplf + \samplg$.
We now state explicit expressions for determining $\Sigma$,  $[\samplfmat]_{kl}$ and $[\samplgvec]_k$.

\begin{result}\label{result3}
	Measurements are distributed optimally when the number of samples for determining individual elements of the matrix and gradient
	are given by
	\begin{align}
		[\samplfmat]_{kl} &=   \epsilon^{-2} \, \Sigma \sqrt{a_{kl} \var \big \{[\qf]_{kl}\big\}}, \\
		[\samplgvec]_{k}  &=  \epsilon^{-2} \, \Sigma \sqrt{b_{k} \var [g_{k}]},
	\end{align}
	respectively. Here $\var[\cdot]$ is the variance of a single measurement of the corresponding element and we
	explicitly define $\Sigma$ via the coefficients $a_{kl}$ and $b_k$ as
	\begin{equation}
		\Sigma := \sum_{k,l=1}^{\nu} \sqrt{a_{kl} \var \big \{[\qf]_{kl}\big\}} + \sum_{k=1}^{\nu} \sqrt{b_k \var[g_k]}.
	\end{equation}
	Furthermore, the symmetry of the Fisher matrix can be explicitly included just by modifying the
	coefficients $a_{kl}$, as discussed in the \summplmat. 
\end{result}
We remark that this result is completely general and can be applied to any of the
metric-aware quantum algorithms \cite{li2017efficient,xiaotheory,samimagtime,koczor2019quantum,quantumnatgrad}.

[Fig. \ref{fig1} (a/c) blue] shows how the optimal distribution of samples
reduces the measurement overhead across the entire evolution
-- most significantly for small regularisation parameters [Fig. \ref{fig1} (a/c), $\eta = 10^{-5}$],
in which case some matrix elements might be crucially larger than others.
Moreover, result~\ref{result3} automatically takes into account the decreasing sampling cost
of the matrix as established in Results~\ref{result1}-\ref{result2}.
This is illustrated in Fig~\ref{fig2} (b);
For the first few iterations, far from convergence, the bulk
of the measurements are directed to the matrix, comparatively few go to the elements of the gradient
[Fig 2 (b), $t=1$]. However, close to convergence, consistent with Result~\ref{result1}, the gradient takes
the majority of the measurements, [Fig \ref{fig2} (b), $t=20$].

\section{Discussion and conclusion}
In this work we established general upper bounds on the sampling cost
of metric-aware variational quantum algorithms (e.g., natural gradient). We analysed how this sampling
cost scales for increasing iterations in Result~\ref{result1} and
for increasing qubit numbers in Result~\ref{result2}.
The latter establishes that the relative measurement cost of the matrix
object $\qf$ is asymptomatically negligible in many practically relevant scenarios,
such as in case of quantum chemistry applications.

Natural gradient has been shown to outperform other optimisation approaches in numerical
simulations~\cite{samimagtime,koczor2019quantum,wierichs2020avoiding}.
We proved in this work that for both an increasing number of iterations and number of qubits
the sampling overhead \emph{per-iteration (per-epoch)} of the natural gradient approach is constant asymptotically
when compared to simple gradient descent.
The most important implication of our results is therefore that
	the \emph{overall cost} of natural gradient is lower since it converges
	to the optimum faster.

We finally established a general technique that optimally distributes
measurements when estimating matrix and vector entries, further reducing the cost of
general metric-aware quantum algorithms. 
Let us finally remark on the generality of our results: our techniques
	are immediately applicable to other problems beyond metric-aware approaches,
	for example, to Hessian-based optimisations via	Eq.~\eqref{errorpropagation}
	as detailed in the Appendix.

\begin{acknowledgments}
Acknowledgements ---
B.\,K. acknowledges funding received from EU H2020-FETFLAG-03-2018 under the grant
agreement No 820495 (AQTION). The authors thank Simon C. Benjamin and Natalia Ares for their
support, stimulating ideas and useful comments on this manuscript.
Numerical simulations in this work used the QuEST and QuESTlink quantum
simulation packages \cite{quest,questlink}.
The authors would like to acknowledge the use of the University of Oxford Advanced
Research Computing (ARC) facility in carrying out this work.
We thank Patrick Coles and Andrew Arrasmith for their useful comments.
\end{acknowledgments}


%

\clearpage
\onecolumngrid

\iffalse
\begin{center}
	\textbf{\large Supplemental Materials: Measurement cost of metric-aware variational quantum algorithms}
\end{center}
\setcounter{equation}{0}
\setcounter{table}{0}
\setcounter{page}{1}
\makeatletter
\renewcommand{\theequation}{S\arabic{equation}}
\renewcommand{\thepage}{S\arabic{page}} 
\renewcommand{\thetable}{S\arabic{table}}  
\renewcommand{\thefigure}{S\arabic{figure}}
\else
\appendix
\fi

\section{Determining variances \label{determineVar}}

\subsection{Pauli decompositions \label{paulidecomps}}
Let us denote the set of Hermitian matrices of dimension $d$ as $\mathrm{Herm}[ \mathds{C}^{d \times d}]$.
The Hamiltonian $\mathcal{H} \in \mathrm{Herm}[ \mathds{C}^{d \times d}]$ of a qubit-system in general decomposes into
a sum over Pauli-operator strings via
\begin{equation} \label{hamildec}
	\mathcal{H} = \sum_{l=1}^{r_h} h_l P_l, \quad \quad \textrm{with} \quad \quad \mathbb{R} \ni h_l:=\tr[\mathcal{H} P_l]/d,
\end{equation}
where $P_l \in \mathrm{Herm}[ \mathds{C}^{d \times d}]$ are tensor products of single-qubit Pauli operators that act on an $N$-qubit system
and form an orthonormal basis of the Hilbert-Schmidt operator space, and $d=2^N$
is the dimensionality. We denote as $r_h \in \mathbb{N}$ the Pauli rank, i.e., the number of non-zero
Pauli components in the Hamiltonian. Note that in general $r_h \leq 4^N$.

In the following derivations we assume for simplicity that ansatz circuits $U_c$ are unitary
and decomposes into a product of individual gates
\begin{equation} \label{circuitEq}
	U_c(\underline{\theta}) = U_\nu(\theta_\nu) \dots  U_2(\theta_2) U_1(\theta_1),
\end{equation}
that typically act on a small subset of the system, e.g., one and two-qubit gates.
We assume in Eq.~\eqref{circuitEq} for ease of notation that each quantum
gate depends on an individual parameter $\theta_i$ with $i=\{1, 2, \dots \nu \}$.

Individual gates $U_k(\theta_k)  \in SU(d)$ of the quantum circuit from Eq.~\eqref{circuitEq} 
are in general of the form $U_k(\theta_k) := \exp[-i \theta_k G_k]$ 
and their generators $G_k \in \mathrm{Herm}[ \mathds{C}^{d \times d}]$ decompose into a sum of Pauli strings
resulting in
\begin{equation*}
	U_k(\theta_k) = \exp[-i \theta_k G_k] = \exp[-i \theta_k \sum_{l=1}^{r^{(k)}_g} g_{kl} P_l], \quad \quad 
	\text{with} \quad \quad
	\mathbb{R} \ni g_{kl}:=\tr[G_k P_l]/d 
\end{equation*}
and $r^{(k)}_g \in \mathbb{N}$ is the Pauli rank of the generator
$G_k$. We additionally assume that $g_{kl} \leq 1/2$ for simplicity -- but any other
upper bound could be specified.
It follows in general that the derivative $\partial_k U_k(\theta_k)$ decomposes
into a sum of $r^{(k)}_g$ unitary operators as
\begin{equation}\label{generalDeriv}
	\partial_k  U_k(\theta_k) = -i \sum_{l=1}^{r^{(k)}_g} g_{kl} P_l U_k(\theta_k).
\end{equation}
For ease of notation, in the following we consider circuits via Eq.~\eqref{circuitEq}
which decompose into gates $U_k(\theta_k)$ with Pauli rank $r_g = 1$.
This is naturally the case for a wide
variety of ansatz circuits, e.g., circuits that consist
of single-qubit rotations and two-qubit $ZZ$ or $XX$ evolution gates as depicted in Fig.~\ref{ansatzfig}.
This assumption results in a
simplified structure of the gates as $U_k(\theta_k) := \exp[-i \theta_k P_k/2]$ 
and their derivatives as
\begin{equation}\label{simpleDeriv}
	\partial_k  U_k(\theta_k) = -\tfrac{i}{2} P_k U_k(\theta_k),
\end{equation}
where $P_k$ is the Pauli generator of the gate $U_k(\theta_k)$.
This construction simplifies our following derivations, however, the generalisation to arbitrary parametrised gates
straightforwardly follows from linearity of Eq.~\eqref{generalDeriv}.

We finally define the partial derivative of the circuit in Eq.~\eqref{circuitEq} using our simplified ansatz as
\begin{equation} \nonumber
	D_k :=  2i \,  \partial_k U_c(\underline{\theta}) =  U_\nu(\theta_\nu) \dots  P_k U_k(\theta_k) \dots  U_2(\theta_2) U_1(\theta_1),
\end{equation}
which itself is unitary via $[D_k]^\dagger = [D_k]^{-1}$ 
(and we omit its explicit dependence on the parameters $\underline{\theta}$)
and $P_l P_l^\dagger = \mathrm{Id}_d$.

We remark that in case of non-unitary parametrisations one would need to consider the general mapping
$\rho(\underline{\theta})  := \Phi(\underline{\theta}) \, \rho_0$. The circuit derivative then decomposes into Pauli terms  as
\begin{equation}
	\partial_k \rho(\underline{\theta})  = \sum_{m,n=1}^{r^{(k)}_p}   p_{kmn} P_m \rho(\underline{\theta}) P_n.
\end{equation}

\subsection{Upper bound on the quantum Fisher information}
We now derive a general upper bound on the quantum Fisher information for unitary parametrisations.
\begin{lemma} \label{qfilemma}
	In case of unitary ansatz circuits that act on arbitrary quantum states $\rho$ via quantum gates
	that decompose into at most $r_g$ Pauli terms, entries of the quantum
	Fisher information matrix are upper bounded as $[\qf]_{kl} \leq r_g^2$.
\end{lemma}
\begin{proof}
	When the ansatz circuit consists of unitary gates, the quantum Fisher information assumes its maximum for pure states.
	Considering the pure state	$\rho=| \psi \rangle \langle  \psi |$, it follows from \cite{koczor2019quantum} that
	\begin{equation*}
		[\qf]_{kl} = 2\, \tr[(\partial_k \rho)(\partial_l \rho)].
	\end{equation*}
	Applying the Cauchy–Schwarz inequality yields
	\begin{equation*}
		2 \tr[(\partial_k \rho)(\partial_l \rho)] \leq 2 \sqrt{ \tr[(\partial_k \rho)(\partial_k \rho)] \, \tr[(\partial_l \rho)(\partial_l \rho)]} \leq  F_{max}
	\end{equation*}
	where $F_{max}$ is a bound on the scalar quantum Fisher information, i.e., diagonal entries of the matrix $\qf$.
	Let us determine this bound via
	\begin{equation} 
		[\qf]_{kk}  = 4 \mathrm{Re}[\langle \partial_k \psi | \partial_k \psi \rangle ]
		- 4|\langle \partial_k \psi | \psi \rangle|^2
		\leq  4 \mathrm{Re}[\langle \partial_k \psi | \partial_k \psi \rangle ] = 4 \langle \partial_k \psi | \partial_k \psi \rangle
	\end{equation}
	for an arbitrary $|  \psi \rangle $.
	It follows from Eq.~\eqref{generalDeriv} that 
	\begin{equation}
		\langle \partial_k \psi | \partial_k \psi \rangle =  \sum_{l,m=1}^{r^{(k)}_g} g_{kl} \, g_{km} \langle  \psi_l | \psi_m \rangle
		\leq (r_g)^2/4,
	\end{equation}
	where $| \psi_m \rangle$ are some valid, normalised states and therefore
	$\langle  \psi_l | \psi_m \rangle \leq 1$ and we used that $g_{kl} \leq 1/2$.
	This finally establishes the general upper bound for unitary ansatz circuits whose gates decompose
	into at most $r_g$ Pauli terms as
	\begin{equation*}
		[\qf]_{kl} \leq r_g^2
	\end{equation*}
	and in case of simplified ans{\"a}tze with $r_g=1$ from Sec.~\ref{paulidecomps} one obtains $[\qf]_{kl} \leq 1$.
	
\end{proof}

\subsection{Components of the gradient \label{gradsec}}

Components of the gradient vector can be measured via Hadamard test.
We discuss this on the example of simplified ans{\"a}tze from Sec.~\ref{paulidecomps},
while the generalisation follows from linearity.
Let us first express the gradient components $g_k:=\partial_k E(\underline{\theta})$
in terms of the derivative circuits from Eq.~\ref{simpleDeriv} as
\begin{align} \nonumber
	g_k  &=  - \im[
	\langle 0 | \, [D_k]^\dagger  \, \mathcal{H} \,U_c \, | 0 \rangle]= - \sum_{l=1}^{r_h} h_l M_{kl},
\end{align}
where the second equation
uses the decomposition of the Hamiltonian into Pauli operators from Eq.~\eqref{hamildec}
via denoting the matrix elements $M_{kl}:=\im\langle 0 | \, [D_k]^\dagger  \, P_l  \,U_c \, | 0 \rangle$.
These matrix elements can be estimated by using an ancilla qubit via the circuits in Fig.~2 of reference \cite{Li2017}
and the corresponding proof can be found in footnote [53] of \cite{Li2017}, refer also to \cite{xiaotheory}. The probability $p$ of measuring this ancilla qubit in the
$|\pm\rangle$ basis with outcome $+1$ determines the matrix elements via
$(2p_{kl} {-}1)=M_{kl}$ for every Pauli component in the Hamiltonian $P_l$.
This finally yields the explicit form of the gradient vector
\begin{equation}
	g_k  =\partial_k E(\underline{\theta}) =- \sum_{l=1}^{r_h} h_l (2p_{kl} {-}1)
\end{equation}
in terms of the measurement probabilities $0 \leq p_{kl} \leq 1$.
Note that each probability $p_{kl} $ is estimated by sampling a binomial
distribution which has a variance $\sigma^2_{kl}  = p_{kl}  (1- p_{kl} )$.
It follows that the variance of the gradient components
are determined by these individual variances via
\begin{equation}\label{gradvariance}
	\var[g_k]  = 4 \sum_{l=1}^{r_h} h_l^2 \, \sigma^2_{kl}  =4 \sum_{l=1}^{r_h} h_l^2 \, p_{kl}  (1- p_{kl} ).
\end{equation}
Re-expressing this variance in terms of the matrix elements via
$p_{kl}= (M_{kl} {+}1)/2$ yields the simplified form
\begin{equation}\label{gradientVar}
	\var[g_k]  = \sum_{l=1}^{r_h} h_l^2 \,   (1{-}[M_{kl}]^2).
\end{equation}
This expression is related directly to the parametrised quantum state 
$| \psi(\underline{\theta}) \rangle $ via the expectation value as
$M_{kl} = - 2 \re\langle \partial_k \psi(\underline{\theta}) | P_l | \psi(\underline{\theta}) \rangle $.
In complete generality, i.e., when gates decompose into a linear
combination of at most $r_g$ Pauli terms, the variance of the gradient entries is upper bounded
(via Eq.~\eqref{gradvariance}) as
\begin{equation} \label{gradVarUB}
	\var[g_k]  \leq r_g \sum_{l=1}^{r_h} h_l^2 =  r_g \hamilnorm,
\end{equation}
where $\hamilnorm $ follows from  the Hilbert-Schmidt scalar
product as
\begin{equation*}
	\hamilnorm := \lVert \mathcal{H} \rVert^2/d :=\tr[\mathcal{H} \mathcal{H}]/d = \sum_{k,l=1}^{r_h} h_k h_l \tr[ P_k  P_l]/d = \sum_{l=1}^{r_h} h_l^2. 
\end{equation*}
via Eq.~\eqref{hamildec} and recall that $\tr[ P_k  P_l] = d\, \delta_{kl}$, where $\delta_{kl}$ is the Kroenecker delta
and $d=2^N$.

So far we have assumed that each term in the Hamiltonian is estimated separately from outcomes
of independent ancilla measurements and the above variance therefore corresponds to overall
$r_h$ measurements. Indeed, advanced techniques could be used for simultaneously measuring commuting terms
in the Hamiltonian (possibly without an ancilla qubit) reducing the overall number of shots \cite{Crawford2019,yen2020measuring,jena2019pauli,gokhale2020n,gokhale2019minimizing,hadfield2020measurements}
and we would like to take this into account in our final result.
We conclude by stating the upper bound on the variance $\var[g_k]$
of a single measurement to estimate the gradient entry $g_k$ as
\begin{equation}\label{Nk_upper_bound}
	\var[g_k] \leq \hamilnorm \fg .
\end{equation}
Here we have introduced the constant factor $\fg$.
We can generally state the bounds $1 \leq \fg \leq r_g r_h$ as $\fg$ depends on the system (type of gates via $r_g$)
and on the measurement technique used for estimating terms in the Hamiltonian (number of commuting groups).
Here the lower bound (best case scenario $ \fg = 1$)
is saturated for Pauli gates ($r_g=1$) and Hamiltonians from Eq.~\eqref{hamildec} in which all terms commute and are measured simultaneously.
The upper bound (worst case scenario) is saturated by Hamiltonians from  Eq.~\eqref{hamildec} in which all $r_h$ terms are estimated
from separate measurements (all terms are non-commuting) and all terms have comparable strengths (optimally
distributing samples does not reduce $\var[g_k]$).
The factor $\fg$ interpolates between these two extremal cases and will correspond to a value in the bounded rage
$1 \leq \fg \leq r_g r_h$. In most of this work we assume a fixed $\mathcal{H}$ and therefore we can treat $\fg$ and $\hamilnorm$
as constants. The only exception is our derivation in Result~2 where we make the mild, general assumption that
the number $r_h$ of terms in the Hamiltonian grows polynomially and therefore necessarily
the product $\hamilnorm \fg = \mathcal{O}(N^b)$ grows in some
polynomial order $b$ with the number of qubits $N$.
To illustrate this, we construct 3 example Hamiltonians in Sec.~\ref{appendix:simulations} and explicitly
compute the polynomial order $b$ in which the cost $\hamilnorm \fg$ of estimating $g_k$ grows with the number of qubits $N$.

Let us now consider mixed quantum states, e.g., due to gate imperfections, via the eigendecomposition $\rho = \sum_{n} p_n |\psi_n \rangle \langle \psi_n | $.
If the parametrisation $\underline{\theta}$ is approximately unitary via $\tfrac{\partial p_n}{\partial \theta_k} \approx 0$,
then gradient components of the expectation value $\tr[ \rho(\underline{\theta}) \mathcal{H} ]$ can be expressed as
\begin{equation}
	\tfrac{\partial }{\partial \theta_k} \tr[ \rho(\underline{\theta}) \mathcal{H} ] 
	\approx \sum_{n} p_n  \tfrac{\partial }{\partial \theta_k}  [\langle \psi_n(\underline{\theta})  | \mathcal{H} |\psi_n(\underline{\theta})  \rangle ]
	= \sum_{n} p_n  [g_{k}]_n
\end{equation}
where $ [g_{k}]_n$ is the gradient that would be measured by the above protocol for the
pure eigenstate $|\psi_n(\underline{\theta})  \rangle$. The above discussed protocol therefore
estimates the correct gradient for mixed states -- as long as the parametrisation is approximately
unitary, such as in case of noisy gates. The same upper bound holds for the variances via $\sum_{n} p_n = 1$ and $0\leq p_n\leq 1$,
and the bound is only saturated by pure states.

In summary, the variance of the gradient entries is upper bounded as $\var[g_k]  \leq   \hamilnorm  \fg$, where
$\fg$ is a constant factor that only depends on the ansatz structure, on the particular quantum algorithm
that is used to estimate the entries and on the Hamiltonian. We remark that the above discussed protocol is used
in other metric-aware quantum algorithms, and our bounds therefore
apply to other vector objects used in these algorithms~\cite{li2017efficient,xiaotheory,samimagtime,koczor2019quantum,quantumnatgrad}.

\subsection{Components of the quantum Fisher information matrix}
We will now focus on determining variances of the quantum Fisher information
entries $[\qf]_{kl} $.
For pure states as $\rho=| \psi \rangle \langle  \psi |$,
entries of the quantum Fisher information can be expressed via the state-vector scalar products \cite{koczor2019quantum}
\begin{equation} \label{qfidef}
	[\qf]_{kl}  = 4 \mathrm{Re}[\langle \partial_k \psi | \partial_l \psi \rangle 
	- \langle \partial_k \psi | \psi \rangle \langle  \psi | \partial_l \psi \rangle],
\end{equation}
The second term in the above equation vanishes when
the global phase evolution of $| \psi \rangle$ is zero \cite{xiaotheory}
and an experimental protocol for measuring the remaining component
$\mathrm{Re}\langle \partial_k \psi | \partial_l \psi \rangle $
was used in \cite{samimagtime} for simulating imaginary time evolution.

We now propose a protocol that determines both
terms in Eq.~\ref{qfidef}.
Assuming the simplified ansatz from Sec.~\ref{paulidecomps},
our protocol allows to evaluate the coefficients by measuring an ancilla qubit 
\begin{align*}
	A_{kl} =  4 \re \langle \partial_k \psi | \partial_l \psi \rangle 
	&= \re \langle 0 | [D_k]^\dagger D_l| 0\rangle 
	= 2[p_a]_{kl} {-}1,\\
	B_k = 2 \re  \langle \partial_k \psi | \psi \rangle
	&=\re \langle 0 | [D_k]^\dagger U_c| 0\rangle 
	= 2[p_b]_k {-}1,\\
	C_k = 2 \im \langle \partial_k \psi | \psi \rangle
	&=\im \langle 0 | [D_k]^\dagger U_c| 0\rangle 
	= 2[p_c]_k {-}1,
\end{align*}
using the circuits in Fig.~2 of reference \cite{Li2017},
refer to footnote [53] of \cite{Li2017} for a proof.
These circuits allow for estimating the probabilities $p_a$, $p_b$ and $p_c$
by sampling the ancilla qubit as a binomial distribution.
The quantum Fisher information is then obtained as
\begin{equation*} 
	[\qf]_{kl}  =  A_{kl} + B_k B_l -  C_kC_l
	= (2[p_a]_{kl} {-}1)+  (2[p_b]_k {-}1) (2[p_b]_l {-}1)
	- (2[p_c]_k {-}1)(2[p_c]_l {-}1).
\end{equation*}
Since the probabilities $p_a$, $p_b$ and $p_c$ are determined from
binomial distributions, their variances are given by, e.g.,
$[\sigma_a^2]_{kl}  = [p_a]_{kl}  (1- [p_a]_{kl})$.
It follows that
\begin{equation*} 
	\var\{ [\qf]_{kl} \} = 4 [\sigma_a^2]_{kl} + 4 [\sigma_b^2]_{k} B_l^2  + 4 [\sigma_b^2]_{l} B_k^2
	+  4[\sigma_c^2]_{k} C_l^2  + 4 [\sigma_c^2]_{l} C_k^2,
\end{equation*}
Substituting $4 [\sigma_b^2]_{k} = (1-[B_l]^2)$ and $4 [\sigma_c^2]_{k} = (1-[C_l]^2)$,
we can express the variances as
\begin{equation*} 
	\var\{ [\qf]_{kl} \} = (1 {-}[A_{kl}]^2) + (1-[B_k]^2) B_l^2  + (1-[B_l]^2) B_k^2 +  (1-[C_k]^2) C_l^2  + (1-[C_l]^2) C_k^2,
\end{equation*}
in terms of the estimated quantities $A_{kl}$, $B_k$ and $C_k$,
and we used the expressions, e.g., $(A_{kl}\ {+}1)/2 = [p_a]_{kl} $.

Note that the inequality $(1-[B_k]^2) B_l^2 \leq 1/4$ is saturated when $B_k=1/\sqrt{2}$ and in general
$ |A_{kl}|, |B_l|, |C_l| \leq 1$. Using this inequality we can  establish the general upper bound
\begin{equation}\label{qfVarUB2}
	\var\{[\qf]_{kl}\}  \leq  2 r_g^2 ,
\end{equation}
when gates decompose into a linear combination of at most $r_g$ Pauli terms.

When assuming noisy unitary circuits, Result~3 in \cite{koczor2019quantum} establishes
that $[\qf]_{kl} \approx 2\, \tr[(\partial_k \rho)(\partial_l \rho)]$
and the approximation becomes exact for pure states as $\rho=| \psi \rangle \langle  \psi |$.
The Hilbert-Schmidt scalar products $\tr[(\partial_k \rho)(\partial_l \rho)]$ can be measured using the 
circuit based on SWAP tests from \cite{xiaotheory} and one can directly estimate the quantity $[\qf]_{kl}  = (2p_{kl}-1)$
by measuring the probability $p_{kl}$ of an ancilla qubit
in case when using the simplified ansatz from Sec.~\ref{paulidecomps}, i.e., 
when gates decompose into single Pauli terms.
We remark that this implementation requires more qubits when compared to the above introduced
pure-state approach. However, it is preferable as it results in negligible approximation
errors when gates are imperfect, refer to \cite{koczor2019quantum}. 
The variance follows as
$\var\{[\qf]_{kl}\}  = 4 p_{kl} (1-p_{kl}) = (1 - [\qf]_{kl}^2) \leq 1$ in case
of the simplified ansatz from Sec.~\ref{paulidecomps} and we have used $[\qf]_{kl} \leq 1$
from Lemma~\ref{qfilemma}.

In complete generality, i.e., when gates decompose into a linear
combination of at most $r_g$ Pauli terms, the variance of the matrix entries is upper bounded as
\begin{equation} \label{qfVarUB1}
	\var\{[\qf]_{kl}\}  \leq r_g^2.
\end{equation}

In summary, the variance of the matrix entries are upper bounded as $\var\{[\qf]_{kl}\}  \leq \fF$, where
$\fF$ is a constant factor that only depends on the ansatz structure and the approach used to estimate the matrix entries.
We remark that the above discussed two protocols are used in other metric-aware quantum algorithms and our bounds therefore
apply to other matrix objects estimated by these algorithms~\cite{li2017efficient,xiaotheory,samimagtime,koczor2019quantum,quantumnatgrad}.

\subsection{Numerical simulations \label{appendix:simulations}}

In our numerical simulations we use the ansatz illustrated in Fig.~\ref{ansatzfig}.
This decomposes into repeated blocks. The first block $B_1$ consists of single-qubit
$X$ rotations while the second block $B_2$ decomposes into nearest-neighbour Pauli $ZZ$ 
gates followed by single qubit $Y$ and $X$ rotations. Each gate depends on an individual parameter
$\theta_k$ with $k\in \{1 \dots \nu\}$. In our numerical simulations we use the ansatz structure
$B_1 B_2 B_2$ which has a linearly growing number of parameters $\nu = \mathcal{O}(N)$ in the number of qubits
via the constant depth $a(N) = \mathcal{O}(N^0)$.

In Fig.~1 we simulate the natural gradient approach for finding the ground state energy of
the spin-chain Hamiltonian
\begin{equation} \label{hamil} 
	\mathcal{H} = 
	\sum_{i=1}^{N-1} J [ 
	\sigma_x^{\{i\}} \sigma_x^{\{i+1\}}
	+ \sigma_y^{\{i\}} \sigma_y^{\{i+1\}} 
	+ \sigma_z^{\{i\}} \sigma_z^{\{i+1\}} 
	]
	+  J[ 
	\sigma_x^{\{1\}} \sigma_x^{\{N\}}
	+ \sigma_y^{\{1\}} \sigma_y^{\{N\}} 
	+ \sigma_z^{\{1\}} \sigma_z^{\{N\}} 
	]
	+\sum_{i=1}^N  \omega_i \, \sigma_z^{\{i\}}.
\end{equation}
which contains identical couplings $xx$, $yy$ and 
$zz$ between nearest neighbours with a constant which
we set $J=1$.
Here $\sigma_\alpha^{\{k\}}$ represent Pauli matrices 
acting on qubit $k$ with $\alpha = \{x,y,z\}$.
We select on-site frequencies $\omega_i$
randomly according to a uniform distribution with values varying between
$-1$ and $1$. The resulting Hamiltonain has a non-trivial, highly
entangled ground state that we aim to approximate using the
(not necessarily optimal) ansatz circuit shown on Fig.~\ref{ansatzfig}.
We initialise the optimisation at a point in parameter space close to the optimum
and we set the step size as $\lambda = 0.2$.

In Fig.~2 we simulate various different Hamiltonians using
the same technique.
In particular, we use Eq.~\eqref{hamil} as the linearly scaling Hamiltonian in Fig.~2 (red).
We define the quadratically scaling Hamiltonian Fig.~2 (blue) as 
\begin{equation} \label{hamil2} 
	\mathcal{H} = 
	\sum_{k>l=1}^{N} J [ 
	\sigma_x^{\{k\}} \sigma_x^{\{l\}}
	+ \sigma_y^{\{k\}} \sigma_y^{\{l\}} 
	+ \sigma_z^{\{k\}} \sigma_z^{\{l\}} 
	]
	+\sum_{k=1}^N  \omega_k \, \sigma_z^{\{k\}},
\end{equation}
while we chose the cubically scaling Hamiltonian Fig.~2 (brown) as
\begin{equation} \label{hamil3} 
	\mathcal{H} = 
	\sum_{l>k}^{N} \sum_{m>l}^{N} J 
	\sigma_x^{\{k\}} \sigma_y^{\{l\}} \sigma_z^{\{m\}}
	+\sum_{k=1}^N  \omega_k \, \sigma_z^{\{k\}}.
\end{equation}
In our simulations we start the optimisation at a random initial
point in parameter space, i.e., $\underline{\theta}$ is selected randomly, and run the optimisation until the gradient
vector is such that $\lVert v\rVert \approx 10^{-1}$. This ensures that the we approximately randomly select
points in parameter space for which the gradient norm is fixed, hence satisfying our assumption in Result~2.
We compute the values of $\samplf$ and $\samplg$ at  $25$ instances of such randomly selected ansatz parameters.
Dots (shading) [solid lines] Fig.~2 shows the average (standard deviation) [fitting] of the ratio $\samplf/\samplg$.

Let us now compute how the product $\hamilnorm \fg =  \mathcal{O}(N^b)$ from Result~2 
(which reflects the cost of estimating a gradient entry $g_k$ via Eq.~\eqref{Nk_upper_bound}) grows with the number of qubits.
Terms in the above Hamiltonians can be grouped into a constant number of commuting groups 
which can be measured simultaneously and therefore  $\fg = \mathcal{O}(1)$ in Eq.~\eqref{Nk_upper_bound}.
Furthermore, the squared sum of the coefficients grows as 
$\hamilnorm = \sum_{l=1}^{r_h} h_l^2 = \mathcal{O}(N^b)$
with $b=1,2,3$, respectively.
We therefore conclude that our upper bound in Eq.~\eqref{Nk_upper_bound} grows as
$\var[g_k] \leq \hamilnorm \fg = \mathcal{O}(N^b)$ and, indeed, 
the product in Result~2 grows as $\hamilnorm \fg = \mathcal{O}(N^b)$
with $b=1,2,3$, respectively.

\begin{figure*}[tb]
	\begin{centering}
		\includegraphics[width=0.7\textwidth]{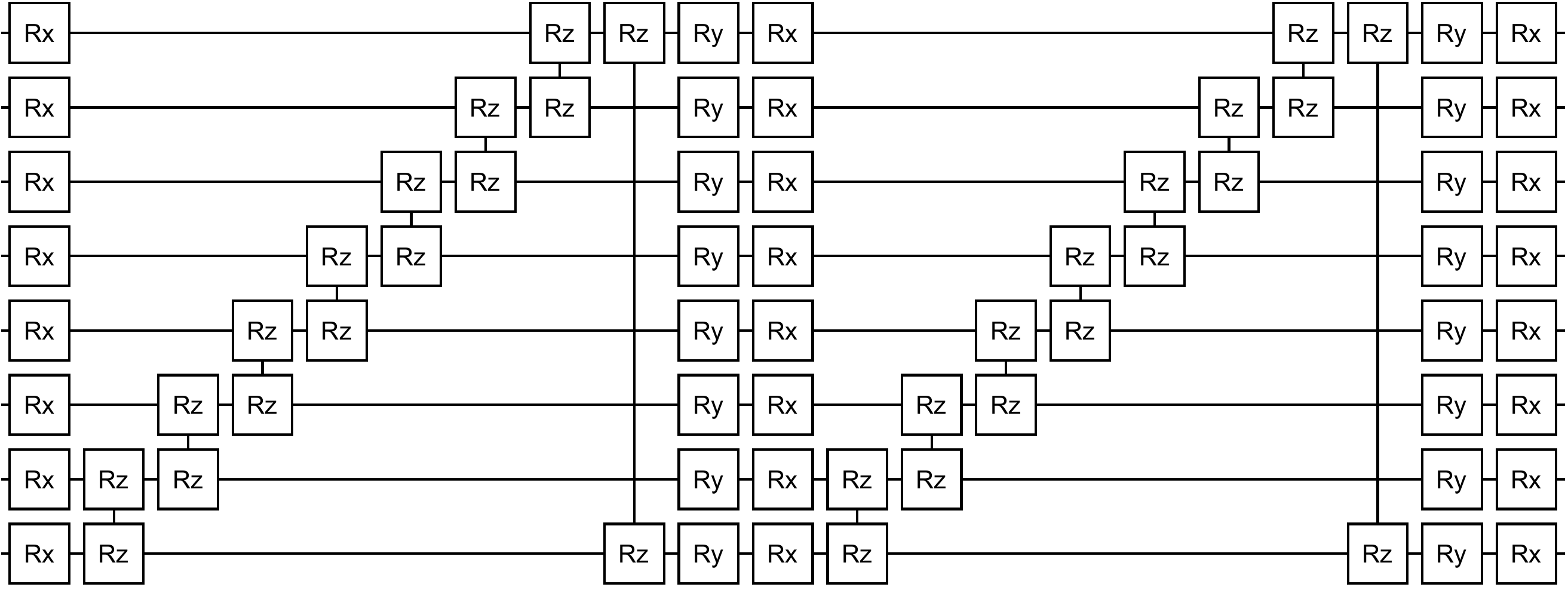}
		\caption{
			Example of an $8$-qubit ansatz structure used in our simulations.
			It consists of repeated blocks of single qubit $X$ and $Y$ rotations and two-qubit
			$ZZ$ evolution gates.
			All gates here have Pauli rank $r_g = 1$ as discussed in Sec.~\ref{paulidecomps}.
			\label{ansatzfig}
		}
	\end{centering}
\end{figure*}

\section{Propagating Variances} \label{sec: PV}

\begin{lemma} \label{propagationlemma}
	Let us define the regularised inverse $\finv := [\qf {+} \eta \mathrm{Id}]^{-1}$
	of the Fisher information matrix for some regularisation parameter $\eta \geq 0$
	and recall that we have defined the error measure $\epsilon^2 := \sum_{k=1}^\nu \var[v_k] $
	with $\underline{v} := \finv \underline{g}$ in the main text.
	If the elements of $\qf$ and $\underline{g}$ are measured independently and their errors are sufficiently small, then the error measure can be written in the form
	\begin{align}
		&\epsilon^{2} = \sum_{k, l = 1}^{\nu} a_{k l} \var \big \{[\qf]_{k l}\big\}  + \sum_{k = 1}^{\nu} b_{k} \var [g_l], \\
		& \textnormal{where} \ \ a_{k l} := \sum_{i,j = 1}^{\nu} [\finv]_{ik}^{2} [\finv]_{l j}^{2} g_{l}^{2}, \ \
		b_{k} := \sum_{l = 1}^{\nu} \big \{[\finv]_{k l} \big \}^{2}.
	\end{align}
\end{lemma}

\begin{proof}\label{proof:em}
	Under the assumption that the elements are measured independently and are sufficiently small, it is appropriate to use the variance formula \cite{Ku1966}, thus we can write error measure in terms of the variance of the elements in $\finv$ and $\underline{g}$ yielding
	\begin{align}
		\epsilon^{2} &= \sum_{k = 1}^{\nu} \var [v_{k}] \\
		&=  \sum_{k, l = 1}^{\nu} \var \big\{[\finv]_{kl} \big\} g_{l}^{2} + \big \{[\finv]_{k l} \big \}^{2} \var [g_l].
	\end{align}
	Now we use the result  derived in \cite{Lefebvre2000} to relate the variance of elements of $\finv$ to elements of $\qf$, namely
	\begin{equation}
		\var \big \{ [\finv]_{k l}\big \} = \sum_{i,j = 1}^{\nu} [\finv]_{ik}^{2} \var \big \{[\qf]_{k l}\big\} [\finv]_{l j}^{2}.
	\end{equation}
	Here we have used that $\var \big \{[\tilde{\mathbf{F}}_Q]_{k l}\big\} = \var \big \{[\qf]_{k l}\big\}$.
	Substituting this result into the error metric and trivially rearranging yields the required result
	
	\begin{align}
		\epsilon^{2} &= \sum_{k, l = 1}^{\nu} \bigg[\sum_{i,j = 1}^{\nu} [\finv]_{ik}^{2} \var \big \{[\qf]_{k l}\big\} [\finv]_{l j}^{2} \bigg] g_{l}^{2} + \big \{[\finv]_{k l} \big \}^{2} \var [g_l] \\
		&= \sum_{k, l = 1}^{\nu} \underbrace{\bigg[\sum_{i,j = 1}^{\nu} [\finv]_{ik}^{2} [\finv]_{l j}^{2} g_{l}^{2} \bigg]}_{a_{k l}} \var \big \{[\qf]_{k l}\big\}  + \sum_{k = 1}^{\nu} \underbrace{\bigg[ \sum_{l = 1}^{\nu} \big \{[\finv]_{k l} \big \}^{2} \bigg]}_{b_{k}} \var [g_k]. \\
	\end{align}
\end{proof}

\subsection{Proof of Theorem~1 \label{proof:theo1}}
\begin{proof}
	Recall that Lemma~\ref{propagationlemma} establishes the error propagation formula
	which we abbreviate as $\epsilon^2 = \epsilon_F^2 + \epsilon_g^2$ via
	\begin{equation}
		\epsilon_F^2 :=  \sum_{\alpha, \beta=1}^\nu a_{\alpha \beta} \var\{[\qf]_{\alpha \beta}\},
		\quad \quad
		\epsilon_g^2 := \sum_{l=1}^\nu b_l \var[g_l].
	\end{equation}
	The coefficients $a_{\alpha \beta}$ can be upper bounded as
	\begin{equation*}
		a_{\alpha \beta} = \sum_{k,l=1}^\nu g_l^2 \,  [\finv]_{k\alpha}^2  [\finv]_{l\beta}^2
		\leq \lVert g \rVert_\infty^2 \sum_{k=1}^\nu [\finv]_{k\alpha}^2  \sum_{l=1}^\nu     [\finv]_{l\beta}^2,
		\quad \quad
		\text{and} \quad \quad
		b_l   = \sum_{k=1}^\nu [\finv]_{k l}^2,
	\end{equation*}
	where $\lVert g \rVert_\infty$ is the absolute largest element in the gradient vector.
	We assume that every matrix and vector element is assigned measurements uniformly as $\samplf/\nu^2$ and
	$\samplg/\nu$ where $\samplf$ and $\samplg$ are the overall number of measurements required to estimate the
	matrix and vector objects such that the vector $\underline{v}$ is obtained to a precision $\epsilon$. 
	Using the upper bounds on the variances of individual gradient vector entries from Eq.~\eqref{gradVarUB}
	and individual matrix entries  from Eq.~\eqref{qfVarUB1} and Eq.~\eqref{qfVarUB2}
	we derive the explicit bound
	\begin{equation*}
		\var\{[\qf]_{\alpha \beta}\} \leq V_F := \nu^2 \samplf^{-1} \, \fF
		\quad \quad \var[g_l] \leq V_G := \nu \samplg^{-1}   \, \hamilnorm \fg,
	\end{equation*}
	where 	$\lVert \mathcal{H} \rVert$ is the Hilbert-Schmidt or Frobenius norm of the Hamiltonian
	and $\fF$, $fg$ are constant factors that depend on the ansatz structure and the
	and the approach used to estimate the gradient/Fisher matrix, refer to Sec.~\ref{paulidecomps}.
	For example for the simplified ansatz ($r_g=1$) in Sec.~\ref{paulidecomps} we obtain $\fF \leq 2$ and 
	$1 \leq \fg \leq r_h$. Here the lower bound $\fg = 1$ is saturated when when all terms in the
	Hamiltonian from Eq.~\eqref{hamildec} commute and can be measured simultaneously while the upper bound $\fg = r_h$ is saturated
	when all $r_h$ terms in the Hamiltonian need to be estimated independently (because they do
	not commute) and their strengths are comparable, refer to Sec.~\ref{gradsec}.

	We use the above derived upper bounds and obtain
	\begin{equation}
		\epsilon_F^2 
		\leq 
		V_F \, \lVert g \rVert_\infty^2 \sum_{\alpha, k =1}^\nu  [\finv]_{k\alpha}^2 \sum_{\beta,l=1}^\nu     [\finv]_{l\beta}^2
		= V_F \lVert g \rVert_\infty^2 \, \lVert \finv \rVert^4,
		\quad \quad \quad
		\epsilon_g^2 \leq
		V_G  \sum_{k,l=1}^\nu [\finv]_{k l}^2 = V_G \, \lVert \finv \rVert^2,
	\end{equation}
	where $\lVert \finv \rVert$ is the Hilbert-Schmidt or Frobenius norm of the inverse matrix $\finv$.

	We now require that $\epsilon^2/2 =: \epsilon_F^2 $ and $ \epsilon^2/2=: \epsilon_g^2$
	as a possible choice to satisfy $\epsilon^2 = \epsilon_F^2 + \epsilon_g^2$.
	This results in the explicit bound on the number of measurements
	after substituting  $V_F$ and $V_G$ as
	\begin{equation*}
		\samplf  \leq  2 \,  \nu^2 \,  \lVert g \rVert_\infty^2 \, \lVert \finv \rVert^4  \epsilon^{-2} \fF \quad \quad 
		\samplg \leq  2 \, \nu\, \lVert \finv \rVert^2  \epsilon^{-2}  \,  \hamilnorm \fg.
	\end{equation*}
	We introduce the notation $\spec[\finv] := \lVert \finv \rVert^2 /\nu = \tfrac{1}{\nu}\sum_{k=1}^\nu \sigma_k^{2}(\finv) $
	to denote the average of the squared singular values of $\finv$. Note that, for example,
	the identity operator yields $\spec[\mathrm{Id}] = 1$ and we derive upper and lower bounds on
	in general in Lemma~\ref{lemmaspec}.
	
	We finally establish the upper bounds
	\begin{equation*}
		\samplf  \leq   2 \,  \nu^4 \,  \lVert g \rVert_\infty^2 \, \spec[\finv]^2  \epsilon^{-2} 	 \, \fF, \quad \quad 
		\samplg \leq 2 \, \nu^2 \, \spec[\finv]  \epsilon^{-2}   \, \hamilnorm \fg.
	\end{equation*}

\end{proof}

\subsection{Proof of Theorem~2 \label{proof:theo2}}
\begin{proof}
	
	Recall that in Sec.~\ref{proof:theo1} we have defined the precision associated with the gradient vector in the natural gradinet
	approach as $\epsilon_g^2 := \sum_{l=1}^\nu b_l \var[g_l]$. We have also defined the total number of measurements $\samplg$ that
	needs 		to be assigned to determining the gradient vector $\underline{g}$ to a precision $\epsilon_g^2$ as
	\begin{equation*}
		\samplg := \frac{\nu}{\epsilon_g^2} \sum_{l=1}^\nu b_l \var[g_l],
	\end{equation*}
	since each gradient entry $g_l$ receives $\samplg/\nu$ samples.

	In the limiting case $\finv \rightarrow \mathrm{Id}$ and $\var\{[\qf]_{k l}\}  \rightarrow 0$ the
	natural gradient approach reduces to the simple gradient descent approach with $b_l=1$.
	We can therefore define the total number of measurements $\samplgd$ required to reconstruct the gradient vector in the simple gradient
	descent approach via $b_l=1$ as 
	\begin{equation*}
		\samplgd := \frac{\nu}{\epsilon_g^2} \sum_{l=1}^\nu \var[g_l].
	\end{equation*}

	Let us start by explicitly writing the ratio of measurements as
	\begin{equation*}
		\frac{\samplg}{\samplgd} = \frac{\sum_{l=1}^\nu	b_l \var[g_l]}{\sum_{l=1}^\nu  \var[g_l]}
	\end{equation*}
	and let us consider the term
	\begin{equation*}
		b_l   = \sum_{k=1}^\nu [\finv]_{k l}^2 = \lVert \mathrm{Col}_l[\finv] \rVert^2
		=  \lVert \finv B_l \rVert^2 \leq  \lVert \finv \rVert_\infty^2  = \sigma_{\mathrm{max}} (\finv)^2 \leq \eta^{-2}
	\end{equation*}
	where $\eta$ is either a regularisation parameter or the smallest singular value of $\qf$,
	$\mathrm{Col}_l[\finv] $ denotes the $l$-th column vector of the matrix $\finv$
	and $B_l$ is the $l$-th standard basis vector with $\lVert B_l \rVert = 1$.
	Our general upper bound follows as
	\begin{equation*}
		\frac{\samplg}{\samplgd} = \frac{\sum_{l=1}^\nu	b_l \var[g_l]}{\sum_{l=1}^\nu  \var[g_l]} \leq 
		\eta^{-2} \frac{\sum_{l=1}^\nu	 \var[g_l]}{\sum_{l=1}^\nu  \var[g_l]} =\eta^{-2}
	\end{equation*}
	We now establish an approximation under the assumption that 
	$\var[g_l]$ does not significantly depend on the index $l$, e.g., when the gradient is vanishing close to an optimal
	point 	 via $M_{kl} \rightarrow 0$ in  Eq.~\eqref{gradientVar} as, e.g, 
	\begin{equation}
		\var[g_k]  = \sum_{l=1}^{r_h} h_l^2 \,   (1{-}[M_{kl}]^2) \rightarrow  \sum_{l=1}^{r_h} h_l^2 = \hamilnorm.
	\end{equation}
	This results in
	\begin{equation*}
		\frac{\samplg}{\samplgd} = \frac{\sum_{l=1}^\nu	b_l \var[g_l]}{\sum_{l=1}^\nu  \var[g_l]}
		\approx \frac{\sum_{l=1}^\nu	b_l }{\nu} = \frac{\sum_{k,l=1}^\nu [\finv]_{k l}^2  }{\nu}
		= \lVert \finv \rVert /\nu = \spec[\finv] .
	\end{equation*}
\end{proof}	
\subsection{Remarks on Theorem~2}
We  establish bounds in case of the relative-precision scheme, i.e., when $\epsilon \propto \lVert g(t) \rVert$
and $\epsilon \propto \lVert v(t) \rVert$ in case of the gradient and natural gradient vectors,
respectively. The upper bound follows via	
\begin{equation}
	\frac{\lVert g \rVert^2}{\lVert v \rVert^2} = 	 	\frac{\lVert g \rVert^2}{\lVert  \finv g \rVert^2} \leq
	\sigma_{\mathrm{min}} (\finv)^{-2}, 	
\end{equation}
and a lower bound can be specified as
\begin{equation}
	\frac{\lVert g \rVert^2}{\lVert v \rVert^2} = 	 	\frac{\lVert g \rVert^2}{\lVert  \qf^{-1} g \rVert^2} \geq
	\sigma_{\mathrm{max}} (\finv)^{-2} 
\end{equation}
and in complete generality
\begin{equation}
	\frac{\samplg}{\samplgd} \frac{\lVert g \rVert^2}{\lVert v \rVert^2}
	\leq  [\sigma_{\mathrm{max}} (\finv)	/\sigma_{\mathrm{min}} (\finv)]^{2} =: \cond[\qf^{-1}]^{2},
\end{equation}
and Lemma~\ref{lemmaspec} establishes that $\cond[\qf^{-1}] \leq \eta^{-1} (\nu r_g + \eta)$.

\begin{lemma} \label{lemmaspec}
	Assuming the simple regularisation $\finv := (\qf + \eta \mathrm{Id}_\nu)^{-1}$,
	the largest singular value of the inverse is upper bounded as $\sigma_\mathrm{max}(\finv) \leq \eta^{-1}$
	and the smallest singular value is lower bounded via $\sigma_\mathrm{min}(\finv) \geq  (\nu r_g + \eta)^{-1}$.
	Moreover, the bounds $ (\nu r_g + \eta)^{-2} \leq \spec[\finv] \leq \eta^{-2}$ and $\cond[\finv] \leq \eta^{-1} (\nu r_g + \eta)$
	hold in general. Here $r_g$ is the largest Pauli rank of the ansatz gates from Sec.~\ref{paulidecomps}.
\end{lemma}
\begin{proof}
	It immediately follows that 
	\begin{equation*}
		\sigma_\mathrm{max}([\qf + \eta \mathrm{Id}_\nu]^{-1}) = [\sigma_\mathrm{min}(\qf + \eta \mathrm{Id}_\nu)]^{-1} \leq \eta^{-1}
	\end{equation*}
	via  $\sigma_\mathrm{min}(\qf + \eta \mathrm{Id}_\nu) \geq \eta$.
	Now we use the boundedness of the matrix elements as $|[\qf ^{-1}]_{kl}| \leq r_g^2$ from Lemma~\ref{qfilemma}
	which establishes the matrix norm $\lVert \qf \rVert_{\mathrm{max}} := \max_{k,l}|[\qf ]_{kl}| \leq r_g^2$.
	This bounds the largest singular value of $\qf$ as
	\begin{equation*}
		r_g^2 \geq	\lVert \qf \rVert_{\mathrm{max}} \geq \lVert \qf \rVert_{\infty} /\nu := \sigma_\mathrm{max}(\qf)/\nu.
	\end{equation*}
	The smallest singular value of the inverse
	is therefore bounded as
	\begin{equation*}
		\sigma_\mathrm{min}([\qf + \eta \mathrm{Id}_\nu]^{-1})  = [\sigma_\mathrm{max}(\qf + \eta \mathrm{Id}_\nu)]^{-1} \geq (\nu r_g^2 + \eta)^{-1}.
	\end{equation*}
	We can now establish the bound
	\begin{equation}
		(\nu r_g^2 + \eta)^{-2} \leq \sigma_\mathrm{min}^2( \finv ) \leq \spec[\finv] \leq  \sigma_\mathrm{max}^2( \finv ) \leq \eta^{-2}
	\end{equation}
	And we can therefore bound the growth rate of the quantity $\spec[\finv]$
	as $\spec[\finv] = \mathcal{O}(\nu^s)$ with $-2 \leq s \leq 0$.
	
\end{proof}

\section{Optimal Measurements}

\begin{lemma}\label{lemma: optimal}
	Measurements are distributed optimally when the number of samples for determining individual elements of the matrix and gradient
	are given by
	\begin{align}
		[\samplfmat]_{kl} &=   \epsilon^{-2} \, \Sigma \sqrt{a_{kl} \var \big \{[\qf]_{kl}\big\}}, \\
		[\samplgvec]_{k}  &=  \epsilon^{-2} \, \Sigma \sqrt{b_{k} \var [g_{k}]},
	\end{align}
	respectively. Here $\var[\cdot]$ is the variance of a single measurement of the corresponding element and we
	explicitly define $\Sigma$ via the coefficients $a_{kl}$ and $b_k$ from Appendix~\ref{sec: PV} as
	\begin{equation}
		\Sigma := \sum_{k,l=1}^{\nu} \sqrt{a_{kl} \var \big \{[\qf]_{kl}\big\}} + \sum_{k=1}^{\nu} \sqrt{b_k \var[g_k]}.
	\end{equation}
\end{lemma}

\begin{proof}
	From Lemma \ref{propagationlemma} we write the error measure as
	
	\begin{equation}
		\epsilon^{2} = \sum_{k, l = 1}^{\nu} a_{k l} \var \big \{[\qf]_{k l}\big\}  + \sum_{k = 1}^{\nu} b_{k} \var [g_k],
	\end{equation}
	where $\var [\cdot]$ denotes the variance in the statistical average over many measurements. Now we allow $\var[\cdot]$ to denote variance in a single measurement while $[\samplfmat]_{kl}$ and $[\samplgvec]_{k}$ are the number of measurement assigned each element $[\qf]_{kl}$ and $g_k$ respectively, so the error measure becomes
	
	\begin{equation}
		\epsilon^{2} = \sum_{k, l = 1}^{\nu} \frac{a_{k l} \var \big \{[\qf]_{k l}\big\}}{[\samplfmat]_{kl}}  + \sum_{k = 1}^{\nu} \frac{b_{k} \var [g_k]}{[\samplgvec]_{k}}.
	\end{equation}
	By minimising error measure, in this form, subject to the constraint of a fixed total number of measurements, so that
	
	\begin{equation}
		N_{opt} = \sum_{k, l = 1}^{\nu} [\samplfmat]_{kl} + \sum_{k = 1}^{\nu} [\samplgvec]_{k},
	\end{equation}
	we find that the optimal fraction of measurement to be assigned to each element is 
	\begin{align}
		\frac{[\samplfmat]_{kl}}{\nopt} =  \frac{\sqrt{a_{kl} \var \big \{[\qf]_{kl}\big\}}}{\Sigma}, \quad
		\frac{[\samplgvec]_{k}}{\nopt} = \frac{\sqrt{b_{k} \var [g_{k}]}}{\Sigma} \\
		\text{where} \quad \Sigma := \sum_{k,l=1}^{\nu} \sqrt{a_{kl} \var \big \{[\qf]_{kl}\big\}} + \sum_{k=1}^{\nu} \sqrt{b_k \var[g_k]}.
	\end{align}
	By substituting this results in the error measure we can remove the dependence on the total number of measurements $N_{opt}$, to yield the required result.
	
\end{proof}

\subsection{Fisher Matrix Symmetry}\label{appendix: symmetry}

\begin{lemma}
	The symmetry of Fisher Matrix can be accounted for by replacing the elements $a_{kl}$ with $a'_{kl}$, where 
	
	\begin{equation}
		a'_{k l} := \begin{cases}
			0 & k < l \\ 
			a_{k k} &  k = l \\
			a_{k l} + a_{l k} & k > l
		\end{cases}
	\end{equation}
\end{lemma}

\begin{proof}
	
	As the Fisher Matrix is symmetric, measurements of $[\qf]_{kl}$ element also constitute measurements of the $[\qf]_{lk}$, so $\var \{[\qf]_{kl}]\} = \var \{[\qf]_{lk} \}$. Thus, the error measure can be written as 
	
	\begin{align}
		\epsilon^{2} = \sum_{k = 1}^{\nu} a_{k k} \var \big \{[\qf]_{k k}\big\} +  \sum_{k > l}^{\nu} 2 a_{kl} \var \big \{[\qf]_{kl}\big\} + \sum_{k = 1}^{\nu} b_{k} \var [g_k].
	\end{align}
	It is possible to force this back into the original form of of the error measure if we define
	\begin{equation}
		a'_{k l} := \begin{cases}
			0 & k < l \\ 
			a_{k k} &  k = l \\
			a_{k l} + a_{l k} & k > l
		\end{cases},
	\end{equation}
	so that the error measure error measure can be written as
	\begin{equation}
		\epsilon^{2} = \sum_{k, l = 1}^{\nu} a'_{k l} \var \big \{[\qf]_{k l}\big\} + \sum_{k = 1}^{\nu} b_{k} \var [g_k].
	\end{equation}
	Using the error measure written in this form as a starting point for the derivation in the proof of Lemma \ref{lemma: optimal} we trivially obtain the same results with the elements $a_{kl}$ replaced with $a'_{kl}$.
	
\end{proof}

\section{Applications beyond natural gradient}
Let us now comment on how main results of this work can be applied to other quantum algorithms beyond natural gradient optimisation.
For this reason, we now consider 3 categories of algorithms and review how the results in the main text can be tailored
to these algorithms.
We note that all algorithms considered in the following use a parameter update rule whereby an inverse
matrix $A^{-1}$ is applied to a vector $v$.

\textbf{Metric-aware optimisation algorithms:}
We have covered metric-aware optimisation algorithms in the main text which include quantum natural gradient descent
and imaginary time evolution \cite{li2017efficient,xiaotheory,samimagtime,koczor2019quantum,quantumnatgrad}.
In this case the matrix object $A$ is the quantum Fisher information, which only depends on the ansatz circuit, 
while the vector object is the gradient vector that depends on both the ansatz circuit and on the Hamiltonian.

\textbf{Variational quantum simulation:}
The time evolution of a quantum system under a Hamiltonian can be simulated using techniques described in
\cite{li2017efficient,xiaotheory}. In such a scenario the matrix object (imaginary part of the quantum geometric tensor)
still only depends on the ansatz circuit and the vector object is related to the gradient vector and therefore all our results apply,
except for Result~\ref{result1}. The reason is the following. In Result~\ref{result1} we 
assumed that when increasing the number of iterations the norm of the gradient vector $v$ vanishes
due to convergence. However, in case of variational simulation, the norm of the vector object does not necessarily decrease.

\textbf{Hessian optimisation:}
Analogously to metric aware optimisations here the inverse of the Hessian matrix is applied to the gradient vector.
Thus our error propagation formula in Eq.~\eqref{errorpropagation} and our optimal measurement distribution scheme in
Result~\ref{result3} immediately apply to this scenario too.
The main difference to the previously discussed scenarios is that the Hessian matrix now depends
on both the ansatz circuit and on the Hamiltonian. For this reason, Theorems~1-2 need to
be modified such that the dependence on the Hamiltonian is taken into account.
As such, the main conclusion of Result~\ref{result1} will still hold: as the optimisation converges
the vanishingly small gradient becomes increasingly more expensive to determine to a sufficient precision.
However, via Result~\ref{result2} it is expected that when we increase the number of qubits the Hessian matrix becomes
increasingly more expensive to estimate due to its dependence on the Hamiltonian.

\end{document}